\documentclass[a4paper,11pt]{article}

\usepackage{graphicx}

\usepackage{fullpage}
\usepackage{libertine}
\usepackage{color}

\usepackage[ruled,vlined]{algorithm2e}
\SetArgSty{textrm}

\usepackage{amsmath,amsfonts,amssymb,amsthm}

\usepackage[breaklinks=true]{hyperref}
\usepackage[svgnames]{xcolor}
\usepackage[capitalise,nameinlink]{cleveref}
\hypersetup{colorlinks={true},linkcolor={DarkBlue},citecolor=[named]{DarkGreen}}

\usepackage{natbib}

\usepackage{subcaption}

\usepackage{tikz}  
\usetikzlibrary{arrows}
\usetikzlibrary{patterns,snakes}
\usetikzlibrary{decorations.shapes}
\tikzstyle{overbrace text style}=[font=\tiny, above, pos=.5, yshift=5pt]
\tikzstyle{overbrace style}=[decorate,decoration={brace,raise=5pt,amplitude=3pt}]
\usetikzlibrary{shapes.geometric}

\newtheorem{theorem}{Theorem}[section]
\newtheorem{corollary}[theorem]{Corollary}

\newtheorem{lemma}[theorem]{Lemma}

\theoremstyle{definition}
\newtheorem*{comment*}{Comment}
\newtheorem{remark}{Remark}[section]

\newcommand{\SC}{\text{SW}}

\newcommand{\favorite}{\text{top}}

\title{\bf Metric Distortion of Obnoxious Distributed Voting}

\usepackage{authblk}
\author{Alexandros A. Voudouris}
\date{School of Computer Science and Electronic Engineering \\ University of Essex, UK}

\begin{document}

\allowdisplaybreaks

\maketitle

\begin{abstract}
We consider a distributed voting problem with a set of agents that are partitioned into disjoint groups and a set of {\em obnoxious} alternatives. Agents and alternatives are represented by points in a metric space. The goal is to compute the alternative that {\em maximizes} the total distance from all agents using a two-step mechanism which, given some information about the distances between agents and alternatives, first chooses a representative alternative for each group of agents, and then declares one of them as the overall winner. Due to the restricted nature of the mechanism and the potentially limited information it has to make its decision, it might not be always possible to choose the optimal alternative. We show tight bounds on the {\em distortion} of different mechanisms depending on the amount of the information they have access to; in particular, we study full-information and ordinal mechanisms.  
\end{abstract}

\section{Introduction} \label{sec:intro}
Inspired by worst-case analysis, the distortion framework was introduced more than 15 years ago as a means of evaluating decision-making mechanisms for social choice problems (such as voting) where limited information (typically of ordinal nature) is available about the preferences of the agents over the alternative outcomes~\citep{distortion-survey}. With only a few exceptions, the vast majority of this literature has focused on settings in which the alternatives are considered {\em desirable}, and thus the goal is to choose one that either maximizes the total utility of the agents or minimizes their total cost. However, there are many important applications where the alternatives are {\em obnoxious}. For example, the alternatives might correspond to different locations for building a factory, in which case the agents would prefer the factory to be built as far away from their houses as possible. Consequently, this implies an underlying optimization problem where the goal is to {\em maximize} the total distance of the agents from the chosen location (rather than minimizing it). Such obnoxious facility location models have been studied through the lens of approximation algorithms~\citep{Tamir91}, operations research \citep{Church2022}, as well as mechanism design~\citep{facility-survey}. 

In this paper, we consider a more general obnoxious voting setting through the lens of the distortion framework. In our model, there are agents that have obnoxious metric preferences over alternatives and are partitioned into disjoint groups. Based on some available information about the preferences of the group members, each group proposes an alternative as a possible winner, and then one of them is chosen to be the final winner. Such distributed models have been considered in the past for the case of desirable alternatives \citep{ratsikas2020distributed,AFV22} aiming to capture district-based elections and other scenarios in which it might be logistically difficult to collect the preferences of the agents in a single pool. In the case of obnoxious alternatives that we focus in this work, such a model naturally captures problems where the goal is to locate obnoxious facilities in a metric space like in the example discussed above. We formalize the setting in the following. 

\subsection{Our Model}
We consider a metric voting problem with a set $N$ of $n$ {\em agents} and a set $A$ of $m$ {\em obnoxious alternatives}.
Agents and alternatives lie in a metric space $d$, that is, the distances between them satisfy the {\em triangle inequality}: $d(x,y) \leq d(x,z) + d(z,y)$, for any $x,y,z \in N \cup A$. 
Following a recent stream of works within the distortion literature~\citep{ratsikas2020distributed,AFV22}, we focus on a {\em distributed} setting in the sense that the agents are partitioned into a set $G$ of $k$ fixed {\em groups}; we denote by $n_g$ the size of group $g \in G$. 

The {\em social welfare} of an alternative $x$ is the total distance of the agents from $x$:
\begin{align*}
    \SC(x) = \sum_{i \in N} d(i,x) = \sum_{g \in G} \sum_{i \in g} d(i,x).
\end{align*}
Our goal is to choose the alternative with {\em maximum social welfare}. This decision is made by a mechanism $M$ that works generally by implementing the following two steps: 
\begin{itemize}
\item[(1)] For every group $g$, given some information about the distances between the agents in $g$ and the alternatives, $M$ chooses an alternative $r_g$ as the {\em representative} of $g$. 
\item[(2)] Given the set $R = \bigcup_g \{r_g\}$ of group representatives, $M$ chooses the final winner $w(M) \in R$. 
\end{itemize}
Clearly, when there is a single group of agents ($k=1$), the setting is {\em centralized}, and the representative of all the agents is the winner. When there are more groups, these two-step mechanisms essentially implement the idea that first each group proposes an alternative (the representative) as a possible winner based on the preferences of its members, and then one of these alternatives is chosen as the final winner. We should also emphasize that any such two-step mechanism satisfies an important locality property according to which the representative of a group must be the same no matter the composition of the other groups; in other words, the representative of a group only depends on the information available for the agents therein and there is no way of distinguishing similar groups based on the whole instance. 

We consider two classes of mechanisms depending on the type of information they use to make decisions: 
\begin{itemize}
    \item A {\em full-information} mechanism has access to the exact distances between agents and alternatives, that is, the mechanism knows the distance $d(i,x)$ for any $i \in N$ and $x \in A$.
    \item An {\em ordinal} mechanism has access to the ordinal preferences of the agents over the alternatives, that is, the mechanism knows a ranking $\succ_i$ for every agent $i$ such that $x \succ_i y$ implies that $d(i,x) \geq d(i,y)$. 
\end{itemize}

Due to the restricted nature of the mechanisms we consider (either due to their two-step decision making process, or due to having access to limited information about the metric space), the chosen alternative might not be the one that maximizes the social welfare. To capture this inefficiency, we adopt the notion of {\em distortion}, which is defined as the worst-case ratio (over all instances with $n$ agents, $m$ alternatives, and $k$ groups) between the maximum possible social welfare and the social welfare achieved by the mechanism:
\begin{align*}
\sup_{(N,A,d,G)} \frac{\max_{x \in A} \SC(x)}{\SC(w(M))}.
\end{align*}
The distortion is by definition at least $1$. Our objective is to identify the best-possible mechanisms with an as small distortion as possible, for each type of information. 

\subsection{Our Contribution}
We show tight bounds on the best possible distortion of {\em deterministic} full-information and ordinal mechanisms. 
We start with the class of full-information mechanisms in Section~\ref{sec:full} for which we show an upper bound of $2\min\{m,k\}-1$ via a mechanism that chooses the representative of each group to be the optimal alternative for the agents therein, and the final winner to be the alternative that represents the most agents. We show a matching lower bound for mechanisms that satisfy a natural assumption that resembles unanimity. While this distortion guarantee grows linearly with $\min\{m,k\}$ in general, it reduces to just $3$ for the important special case of a line metric since all alternatives besides just two can be eliminated by exploiting the geometry of the line. We complement this result with an unconditional lower bound of $3$ on the distortion of full-information mechanisms, thus completely resolving this fundamental case. 

We next turn our attention to the class of ordinal mechanisms in Section~\ref{sec:ordinal} and first consider the centralized setting with a single group of agents. We show a tight bound of $3$, which follows by the same mechanisms that achieve the best possible distortion of $3$ for the classic setting with desirable alternatives~\citep{gkatzelis2020resolving,kempe2022veto}. For general distributed instances with multiple groups of agents, we show a tight bound of $4\min\{m,k\}-1$ via a mechanism that chooses the representative of each group using the best-possible centralized mechanism (of distortion $3$), and the final winner to be the alternative that represents the most agents. We finally show a refined tight bound of $7$ for the line metric using similar techniques as in the case of full-information mechanisms. 

\subsection{Related Work}
Since its introduction by \citet{procaccia2006distortion}, 
the distortion has been studied for many different versions of fundamental social choice problems, 
including single-winner voting~\citep{anshelevich2018approximating,boutilier2015optimal,charikar24breaking,ebadian2022optimized,gkatzelis2020resolving,kempe2022veto}, 
multi-winner voting~\citep{caragiannis2017subset,CSV22,jaworski2020committees},
participatory budgeting~\citep{benade2021participatory}, 
facility location~\citep{feldman2016voting}, 
matching~\citep{anari2023matching,amanatidis2022matching,Amanatidis2024dice,filos2014RP,latifian2024approval}, 
and other optimization problems~\citep{anshelevich2016blind,burkhardt2024low}.

The distributed voting model, in which the agents are partitioned into disjoint groups and decisions are made by two-step mechanisms, was first considered by \citet{ratsikas2020distributed} for utilitarian voting with agents that have normalized values for desirable alternatives. More related to our work here are the papers of \citet{AFV22} and \citet{voudouris2023tight} who studied the problem for agents with metric preferences, but again for desirable alternatives. The authors showed an array of small constant bounds on the distortion of full-information and ordinal mechanisms for various social objectives. With the notable exception of the line metric, the distortion in our obnoxious model turns out to be a function of the number of alternatives and the number of groups, even when full information is available within the groups. The distributed model has also been considered for strategyproof facility location on a line~\citep{filos2024distributedFL}. The distortion of obnoxious metric voting has also been studied before by \citet{fotakis2022aligned} for the centralized setting and particular metric spaces with aligned alternatives. In contrast, we here consider a more general distributed setting (which includes the centralized one as a special case) and show bounds even for general metrics. 

\section{Full-information Mechanisms} \label{sec:full}
In this section we will show bounds on the distortion of full-information mechanisms which are given access to the distances between agents and alternatives in the metric space. 

\subsection{Upper Bounds}
Given this type of information, we can identify the alternative that maximizes the total distance from the agents within each group (and thus achieve distortion $1$ when $k=1$). Using this, for instances with $k \geq 2$ groups, we consider the {\sc Max-Weight-of-Optimal} mechanism, which works as follows: 
The representative of each group is the optimal alternative for the agents in the group.
The final winner is the alternative that represents the most agents (i.e., an alternative that maximizes the total size of the groups it represents).
See Mechanism~\ref{mech:full} for a description of the mechanism using pseudocode.

\newcommand\mycommfont[1]{\normalfont\textcolor{blue}{#1}}
\SetCommentSty{mycommfont}
\begin{algorithm}[th]
\SetNoFillComment
\caption{\sc Max-Weight-of-Optimal}
\label{mech:full}
{\bf Input:} Distances $d(i,x)$ for every $i \in N$ and $x \in A$\;
{\bf Output:} Winner $w$\;

\For{each $x \in A$}{
    $G_x \gets \varnothing$\;
}

\For{each group $g \in G$}{
    $r_g \gets \arg\max_{y \in A} \sum_{i \in g} d(i,y)$\;
    $G_{r_g} \gets G_{r_g} \cup \{g\}$\;
}
$w \gets \arg\max_{x \in A} \sum_{g \in G_x} n_g$\;
\end{algorithm}

\begin{theorem}\label{thm:unrestricted}
The distortion of {\sc Max-Weight-of-Optimal} is at most $2\min\{m,k\}-1$.
\end{theorem}

\begin{proof}
Let $w$ be the alternative chosen as the winner by the mechanism, and $o$ an optimal alternative. 
Let $R$ be the set of representatives, and denote by $G_x$ the set of groups that $x \in R$ represents; 
note that $|R| \leq  \min\{k,m\}$. 
Since the representative $r_g$ of group $g$ is the optimal alternative for the agents in $g$, we have
\begin{align}\label{eq:full:optimal-in-groups}
    \forall x \in A: \sum_{i \in g} d(i,x) \leq \sum_{i \in g} d(i,r_g). 
\end{align}
Also, by the definition of the mechanism, $w$ maximizes the total size of the groups it represents, and thus
$\sum_{g \in G_w} n_g \geq \sum_{g \in G_x} n_g$ for every $x \in R$. 
By adding these inequalities, we have
\begin{align}\label{eq:max-weight}
    (|R|-1) \sum_{g \in G_w} n_g \geq \sum_{g \not\in G_w} n_g.
\end{align}
Let $r \in \arg\max_{y \in R} d(w,y)$ be a most-distant group representative from $w$. 
Using \eqref{eq:full:optimal-in-groups} with $x=o$, and the triangle inequality, we can upper-bound the optimal social welfare:
\begin{align*}
\SC(o) = \sum_g \sum_{i \in N_g} d(i,o) 
&\leq \sum_{g \in G} \sum_{i \in N_g} d(i,r_g) \\
&\leq \sum_{g \in G} \sum_{i \in N_g} \bigg( d(i,w) + d(w,r_g) \bigg) \\
&\leq  \SC(w) + d(w,r) \cdot \sum_{g\not\in G_w} n_g.
\end{align*}
We now bound the second term using \eqref{eq:max-weight}, the triangle inequality, and \eqref{eq:full:optimal-in-groups} with $x=r$:
\begin{align*}
d(w,r) \cdot \sum_{g\not\in G_w} n_g 
&\leq d(w,r) \cdot (|R|-1) \sum_{g\in G_w} n_g \\
&= (|R|-1) \sum_{g\in G_w} \sum_{i \in g} d(w,r) \\
&\leq (|R|-1) \sum_{g\in G_w} \sum_{i \in g} \bigg( d(i,w) + d(i,r) \bigg) \\
&\leq (|R|-1) \bigg( \SC(w) + \sum_{g\in G_w} \sum_{i \in g} d(i,w) \bigg) \\
&=  2(|R|-1) \cdot \SC(w). 
\end{align*}
Putting everything together, we have
\begin{align*}
    \SC(o) \leq \bigg( 1 + 2(|R|-1) \bigg) \cdot \SC(w) = (2|R|-1) \cdot \SC(w). 
\end{align*}
The theorem follows since $|R| \leq \min\{m,k\}$.
\end{proof}

By Theorem~\ref{thm:unrestricted}, we obtain distortion upper bounds for several cases, such as when the number of alternatives $m$ or the number of groups $k$ is small. The case $m=2$ is particularly interesting since it also implies an upper bound of $3$ for the line metric. By the geometry of the line, it is not hard to observe that the optimal alternative, which maximizes the total distance from the agents, is either the leftmost or the rightmost alternative. Hence, all other alternatives can be eliminated, and an upper bound of $3$ is directly obtained by setting $m=2$. 

\begin{corollary}
When there are only two alternatives or when the metric space is a line, the distortion of {\sc Max-Weight-of-Optimal} is at most $3$. 
\end{corollary}

\subsection{Lower Bounds}
We now show lower bounds on the distortion of full-information mechanisms. 
For general metrics, we show a lower bound of $2\min\{m,k\}-1$ for the class of {\em group-unanimous} mechanisms that satisfy the following property: For any group in which {\em all} agents are co-located, the representative of the group is chosen to be the alternative that is farthest from the agents. This property is satisfied by most natural mechanisms, including {\sc Max-Weight-of-Optimal}, but is not true in general. 

\begin{theorem}\label{thm:lower:unrestricted:general:tie}
The distortion of any full-information group-unanimous mechanism is at least $2\min\{m,k\}-1$.
\end{theorem}

\begin{proof}
We consider instances with $m=k$ alternatives $\{x_1,\ldots,x_k\}$ that are located in the metric space so that they are all at distance $1$ from each other. Given that the symmetry between the alternatives, for any instance consisting of groups that are of the same size $\lambda$ and are represented by different alternatives (in particular, group $g_j$ is represented by alternative $x_j$ for any $j \in [k]$), we can assume without loss of generality that the winner is the representative of the last group. 
Now consider a specific instance with the following groups:
\begin{itemize}
    \item For $j \in [k-1]$, group $g_j$ consists of $\lambda$ agents that are located near alternative $x_k$ such that $x_j$ is uniquely the farthest alternative.
    \item Group $g_k$ consists of $\lambda$ agents that are located at distance nearly $1/2$ from all alternatives, such that $x_k$ is the farthest alternative.
\end{itemize}
Due to the structure of the groups and group-unanimity, the representative of group $g_j$ is $x_j$ for any $j \in [k]$, and thus the overall winner is $x_k$. Since $\SC(x_k) \approx \lambda/2$ and $\SC(x_j) \approx (k-1)\lambda + \lambda/2$, we obtain a lower bound of  $2k-1$.
\end{proof}

\begin{remark}\label{remark}
The proof of Theorem~\ref{thm:lower:unrestricted:general:tie} relies on an instance with $m=k$ and leads to a lower bound of $2k-1$. It is not hard to observe that the lower bound for any $m \geq k$ by extending the instance to include extra dummy alternatives that are not used. To obtain a lower bound for $m < k$, we can focus on the case where $k$ is a multiple of $m$ and modify the instance in the proof of Theorem~\ref{thm:lower:unrestricted:general:tie} so that there are $k/m$ groups (rather than just one) represented by alternative $x_j$, for any $j \in [m]$; this leads to a bound of $2m-1$. 
\hfill $\qed$
\end{remark}

We next show an unconditional tight lower bound of $3$ on the distortion of mechanisms for when there are only two alternatives and the metric is a line. 

\begin{theorem}\label{thm:unrestricted:lower:3:wlog}
The distortion of any full-information mechanism is at least $3$, even when the metric space is a line and there are only two alternatives.     
\end{theorem}

The rest of this section is dedicated to proving this theorem. Our methodology is based on induction and is similar to the methodology of \citet{AFV22} for showing lower bounds for the line metric in the desirable setting. Throughout the proof, we consider instances with two alternatives located at $0$ and $1$ on the line of real numbers. Without loss of generality, we assume that, if there are just two groups and each of them has a different representative (that is, one is represented by alternative $0$ while the other is represented by alternative $1$), then the winner is $0$. We will now argue that any low-distortion mechanism must choose a particular representative for the group $g_x$ in which all agents are located at $x \in [0,1]$. To simplify our calculations in the following, we assume that there is just one agent in such groups (where all agents are co-located). 

\begin{lemma}\label{lem:unrestricted:small-or-large-x}
Any full-information mechanism with distortion strictly smaller than $3$ 
must choose the representative of group $g_x$ to be 
alternative $0$ for $x \in [3/4,1]$
and alternative $1$ for  $x \in [0,1/4]$.
\end{lemma}

\begin{proof}
Let $x\in [3/4,1]$; the case $x \in [0,1/4]$ is symmetric. Suppose otherwise that the representative of $g_x$ is $1$ instead of $0$. Then, in any instance consisting of copies of $g_x$ (which means that all agents are located at $x$), 
$1$ is the only representative, and thus the overall winner. 
Since $\SC(0) = n x$ and $\SC(1) = n(1-x)$, the distortion is at least $x/(1-x) \geq 3$, a contradiction.
\end{proof}

\begin{lemma}\label{lem:unrestricted:majority}
Suppose a full-information mechanism with distortion strictly smaller than $3$ which chooses alternative $1$ as the winner when there are $2\ell+1$ groups such that $\ell$ of them are represented by $0$ while the remaining $\ell+1$ are represented by $1$, for any $\ell \geq 1$. Then, the mechanism must choose alternative $0$ as the representative of group $g_{(2\ell+3)/(4\ell+4)}$.
\end{lemma}

\begin{proof}
Suppose otherwise that $1$ is chosen as the representative of group $g_{(2\ell+3)/(4\ell+4)}$. 
Then, consider an instance consisting of $2\ell+1$ groups such that there are $\ell$ copies of $g_1$ and $\ell+1$ copies of $g_{(2\ell+3)/(4\ell+4)}$. By Lemma~\ref{lem:unrestricted:small-or-large-x}, the representative of $g_1$ must be $0$, while the representative of $g_{(2\ell+3)/(4\ell+4)}$ is $1$ by assumption. Hence, we have an instance in which there are $\ell$ groups represented by $0$ and $\ell+1$ represented by $1$, leading to $1$ being chosen as the winner by the mechanism due to the assumption of the statement. 
However, 
\begin{align*}
    \SC(1) = (\ell+1) \bigg( 1- \frac{2\ell+3}{4\ell+4} \bigg) = \frac{2\ell+1}{4},
\end{align*}
whereas
\begin{align*}
    \SC(0) = \ell + (\ell+1) \cdot \frac{2\ell+3}{4\ell+4} = \frac{6\ell+3}{4},
\end{align*}
leading to a distortion of $3$, a contradiction. 
\end{proof}

Using Lemmas~\ref{lem:unrestricted:small-or-large-x} and~\ref{lem:unrestricted:majority}, we can now argue that the assumption of Lemma~\ref{lem:unrestricted:majority} is true for any low-distortion mechanism. 

\begin{lemma}\label{lem:unrestricted:g1/2+}
Any full-information mechanism with distortion strictly smaller than $3$ must choose alternative $1$ as the winner when there are $2\ell+1$ groups such that $\ell$ of them are represented by $0$ while the remaining $\ell+1$ are represented by $1$, for any $\ell \geq 1$.
\end{lemma}

\begin{proof}
We will prove the statement by induction using Lemmas~\ref{lem:unrestricted:small-or-large-x} and~\ref{lem:unrestricted:majority} repeatedly. 

\medskip
\noindent
{\bf Base case: $\ell=1$.}
Suppose otherwise that the mechanism chooses $0$ as the winner when there are $3$ groups such that one of them is represented by $0$ while the other two are represented by $1$. Consider an instance consisting of $g_{3/4}$ and two copies of $g_0$. By Lemma~\ref{lem:unrestricted:small-or-large-x}, the representative of $g_{3/4}$ must be $0$ (using $x=3/4$), while the representative of $g_0$ must be $1$ (using $x=0$). Hence, the overall winner is $0$. However, $\SC(0) = 3/4$ and $\SC(1) = 1/4 + 2 = 9/4$, leading to a distortion of $3$, a contradiction. 

\medskip
\noindent 
{\bf Hypothesis:} We assume that the statement is true for $\ell-1$, that is, the winner is $1$ when there are $2(\ell-1)+1 = 2\ell-1$ groups such that $\ell-1$ of them are represented by $0$ while the remaining $\ell$ are represented by $1$. Hence, by Lemma~\ref{lem:unrestricted:majority}, the mechanism must choose $0$ as the representative of group $g_{(2(\ell-1)+3)/(4(\ell-1)+4)} = g_{(2\ell+1)/(4\ell)}$.

\medskip
\noindent 
{\bf Induction step:}
Suppose otherwise that the mechanism chooses $0$ as the winner when there are $2\ell+1$ groups such that $\ell$ of them are represented by $0$ while the remaining $\ell+1$ are represented by $1$. Consider an instance consisting of 
$\ell$ copies of $g_{(2\ell+1)/(4\ell)}$ which are represented by $0$ due to our induction hypothesis, 
and $\ell+1$ copies of $g_0$ which are represented by $1$ by Lemma~\ref{lem:unrestricted:small-or-large-x}. 
Since
\begin{align*}
    \SC(0) = \ell \cdot \frac{2\ell+1}{4\ell} = \frac{2\ell+1}{4}
\end{align*}
and
\begin{align*}
    \SC(1) = \ell \cdot \bigg( 1 -  \frac{2\ell+1}{4\ell} \bigg) + \ell+1 = \frac{6\ell+3}{4},
\end{align*}
the distortion is at least $3$, a contradiction. 
\end{proof}

We are now ready to prove Theorem~\ref{thm:unrestricted:lower:3:wlog}. 

\begin{proof}[Proof of Theorem~\ref{thm:unrestricted:lower:3:wlog}]
By Lemma~\ref{lem:unrestricted:g1/2+} and Lemma~\ref{lem:unrestricted:majority}, we have that any mechanism with distortion strictly smaller than $3$ must choose $0$ as the representative of group $g_{(2\ell+3)/(4\ell+4)}$ for any $\ell \geq 1$. Taking $\ell$ to infinity, we have that the mechanism chooses $0$ as the representative of group $g_{1/2+\varepsilon}$ for some infinitesimal $\varepsilon > 0$. 

Now consider an instance that consists of $g_{1/2+\varepsilon}$ and $g_0$. Since the former is represented by $0$ and the latter is represented by $1$ (due to Lemma~\ref{lem:unrestricted:small-or-large-x}), the winner is $0$. However, since $\SC(0) \approx 1/2$ and $\SC(1) \approx 3/2$, the distortion is at least $3$. 
\end{proof}

\section{Ordinal Mechanisms} \label{sec:ordinal}
We now turn our attention to ordinal mechanisms which are given access to the ordinal preferences of the agents over the alternatives. 

\subsection{Upper Bounds}
We start with the centralized setting, in which there is a single group of agents, and show that it is possible to achieve a tight bound of $3$. As in the classic setting where the alternatives are desirable (rather than obnoxious), the bound is achieved by choosing any alternative whose domination graph attains a perfect matching. Formally, the domination graph of an alternative $x$ is a bipartite graph such that each side consists of the set of agents and a directed edge from agent $i$ to agent $j$ exists if and only if $i$ prefers $x$ over the favorite alternative $\favorite(j)$ of $j$. There are several centralized mechanisms that, given the ordinal preferences of the agents, output an alternative with this ordinal property, for example, {\sc Plurality-Matching}~\citep{gkatzelis2020resolving} and the much simpler {\sc Plurality-Veto}~\citep{kempe2022veto}. 
In the obnoxious setting that we consider, since alternatives are ordered in decreasing distance, the existence of a perfect matching $\mu$ in the domination graph of an alternative $x$ implies that  $d(i,x) \geq d(i,\favorite(\mu(i)))$ for every agent $i$. 

\begin{theorem}\label{thm:ordinal-centralized}
For $k=1$, the distortion of any ordinal mechanism that chooses an alternative whose domination graph attains a perfect matching is at most $3$, and this is tight over all ordinal mechanisms. 
\end{theorem}

\begin{proof}
Let $w$ be the alternative chosen as the winner by the mechanism, and $o$ an optimal alternative.
Since the domination graph of $w$ attains a perfect matching $\mu$, we have that $d(i,w) \geq d(i,\favorite(\mu(i)))$ for every agent $i$.
Also, $d(i,\favorite(i)) \geq d(i,o)$ by definition.
Using these properties and by applying the triangle inequality, we obtain
\begin{align*}
    \SC(o) 
    &\leq \sum_{i \in N} d(\mu(i),\favorite(\mu(i))) \\
    &\leq \sum_{i \in N} \bigg( d(\mu(i),w) + d(i,w) + d(i, \favorite(\mu(i))) \bigg) \\
    &\leq \sum_{i \in N} \bigg( d(\mu(i),w) + 2\cdot d(i,w) \bigg) \\
    &= 3 \cdot \SC(w).
\end{align*}

For the lower bound, consider an instance with two alternatives $a$ and $b$. Half of the agents prefer $a$ and the other half prefer $b$. Given these ordinal preferences, any of the two alternatives can be chosen as the winner, say $a$. Consider now the following consistent positions of alternatives and agents on the line of real numbers: 
$a$ is at $0$, $b$ is at $2$, the agents that prefer $a$ are at $1$, and the agents that prefer $b$ are at $0$. Hence, $\SC(a) = n/2$ and $\SC(b) = 3n/2$, leading to a distortion lower bound of $3$ on the distortion of any ordinal mechanism. 
\end{proof}

For instances with $k \geq 2$ groups, we consider the {\sc Max-Weight-of-Domination} mechanism that works as follows: The representative of each group is chosen to be any alternative whose domination graph attains a perfect matching (given the preferences of the agents in the group), and then the final winner is the alternative that represents the most agents. See Mechanism~\ref{mech:ordinal}. 

\SetCommentSty{mycommfont}
\begin{algorithm}[th]
\SetNoFillComment
\caption{\sc Max-Weight-of-Domination}
\label{mech:ordinal}
{\bf Input:} Ordinal preference $\succ_i$ of every agent $i \in N$\;
{\bf Output:} Winner $w$\;

\For{each $x \in A$}{
    $G_x \gets \varnothing$\;
}

\For{each group $g \in G$}{
    $r_g \gets $ {\sc Plurality-Veto}$(\{\succ_i\}_{i \in g})$\;
    $G_{r_g} \gets G_{r_g} \cup \{g\}$\;
}
$w \gets \arg\max_{x \in A} \sum_{g \in G_x} n_g$\;
\end{algorithm}

\begin{theorem} \label{thm:distributed-perfect-matching-4q-1}
The distortion of {\sc Max-Weight-of-Domination} is at most $4\min\{m,k\}-1$. 
\end{theorem}

\begin{proof}
Let $w$ be the alternative chosen by the mechanism, and $o$ an optimal alternative. Let $R$ be the set of alternatives that represent some group, and denote by $G_x$ the set of group that $x\in R$ represents; note that $|R| \leq \min\{m,k\}$. 
Since the domination graph of the representative $r_g$ of group $g$ attains a perfect matching $\mu_g$ for the agents in $g$, we have 
\begin{align} \label{eq:domination:property}
    \forall g \in G, i \in g: d(i,\favorite(\mu_g(i)) \leq d(i,r_g).
\end{align}
By Theorem~\ref{thm:ordinal-centralized}, we also have that
\begin{align} \label{eq:domination:distortion}
    \forall g \in G, \forall x \in A:  \sum_{i \in g} d(i,x) \leq 3\cdot \sum_{i \in g} d(i,r_g).
\end{align}
Also, since $w$ maximizes the total size of the groups it represents, $\sum_{g \in G_w} n_g \geq \sum_{g \in G_x} n_g$ for every $x \in R$, and thus 
\begin{align} \label{eq:ordinal:max-weight}
    (|R|-1) \sum_{g \in G_w} n_g \geq \sum_{g \not\in G_w} n_g.
\end{align}
Since by definition $d(i,\favorite(i)) \geq d(i,o)$ for any agent $i$, using the triangle inequality, we can bound the optimal social welfare as follows:
\begin{align*}
\SC(o) = \sum_g \sum_{i \in g} d(i,o) 
&\leq \sum_{g \in G} \sum_{i \in g} d(\mu_g(i),\favorite(\mu_g(i)))\\ 
&\leq \sum_{g in G} \sum_{i \in g} \bigg( d(\mu_g(i),w) + d(i,w) + d(i,\favorite(\mu_g(i))) \bigg) \\
&= 2\cdot \SC(w) + \sum_{g \in G} \sum_{i \in g} d(i,\favorite(\mu_g(i))).
\end{align*}
We now bound the second term.
Let $r \in \arg\max_{y \in R} d(w,y)$ be a most-distant group representative from $w$. 
Using \eqref{eq:domination:property}, the triangle inequality, the definition of $r$, \eqref{eq:ordinal:max-weight}, the triangle inequality again, and finally \eqref{eq:domination:distortion} with $x=r$, we obtain:
\begin{align*}
\sum_{g\in G} \sum_{i \in g} d(i,\favorite(M_g(i))) 
&\leq \sum_{g \in G} \sum_{i \in g} d(i,r_g) \\
&\leq\sum_{g \in G} \sum_{i \in g} \bigg( d(i,w) + d(w,r_g) \bigg) \\
&= \SC(w) + \sum_g \sum_{i \in g}  d(w,r_g) \\
&\leq \SC(w) +  d(w,r) \cdot \sum_{g\not\in G_w} n_g \\
&\leq \SC(w) +  d(w,r) \cdot (|R|-1) \sum_{g\in G_w} n_g \\
&= \SC(w) + (|R|-1) \sum_{g\in G_w} \sum_{i \in N_g} d(w,r) \\
&\leq \SC(w) + (|R|-1) \sum_{g\in G_w} \sum_{i \in N_g} \bigg( d(i,w) + d(i,r) \bigg) \\
&\leq \SC(w) + 4\cdot (|R|-1) \cdot \SC(w) \\
&= (4|R|-3) \cdot \SC(w).
\end{align*}
Putting everything together, we have
\begin{align*}
\SC(o) \leq (4|R|-1) \cdot \SC(w),
\end{align*}
and the bound follows by the fact that $|R| \leq \min\{m,k\}$.
\end{proof}

As in the case of full-information mechanisms, when the metric is a line, if we are also given ordinal information about the alternatives, we can eliminate all alternatives besides the leftmost and rightmost ones, one of which is the optimal. Hence, we immediately obtain an upper bound $7$ by running {\sc Max-Weight-of-Domination} on the remaining two alternatives.

\begin{corollary}
When there are only two alternatives or when the metric space is a line, the distortion of {\sc Max-Weight-of-Domination} is at most $7$. 
\end{corollary}

\subsection{Lower Bounds}
Next, we show a tight lower bound on the distortion of ordinal mechanisms. We first need the following lemma to argue about the alternatives that can be selected as the representative of a group based on the top preferences of the agents in the group. 

\begin{lemma} \label{lem:lower:ordinal:top}
Any ordinal distributed mechanism with finite distortion must choose as the representative of a group one of the alternatives that is top-ranked by some agent in the group. 
\end{lemma}

\begin{proof}
Suppose towards a contradiction that there is an ordinal mechanism with finite distortion which chooses alternative 
$y \not\in \bigcup_{i \in g} \favorite(i)$ as the representative of a group $g$. 
So, in any instance consisting of only copies of this group, the only representative alternative, and thus the winner, is $y$. 
Now consider the following line metric that is consistent with the top preferences of the agents:
\begin{itemize}
    \item Every alternative $x \in \bigcup_i \favorite(i)$ is at $0$;
    \item All the remaining alternatives as well as all agents are at $1$. 
\end{itemize}
Hence, $\SC(y) = 0$ while $\SC(x) = n$ for any $x \in \bigcup_{i \in g} \favorite(i)$, leading to infinite distortion.
\end{proof}

Using this lemma, we can show a lower bound of $4\min\{m,k\}-1$, which holds even when the distances between alternatives are known. This further shows that {\sc Max-Weight-of-Domination} is the best possible ordinal mechanism for general instances. 

\begin{theorem}\label{thm:lower:ordinal:general}
For any $m \geq 3$, the distortion of any ordinal distributed mechanism is $4\min\{m,k\}-1$, even when the locations of the alternatives are known. 
\end{theorem}

\begin{proof}
We consider an instance with $m=k+1$ alternatives $a,b,x_1,\ldots,x_{k-1}$ that are located in the metric space so that they are all at distance $1$ from each other.\footnote{As discussed in Remark~\ref{remark} for the proof of Theorem~\ref{thm:lower:unrestricted:general:tie}, we could also have more alternatives in this proof as well without any change in the arguments. We could also design an instance with some $m < k$ by having multiple groups represented by the same alternative.} 
Given that all alternatives are symmetric, for any instance such that the groups are of the same size $\lambda \geq 1$ and are represented by different alternatives, we can assume without loss of generality that the winner is the representative of the last group. 

Consider now the following instance in which, for $j \in [k-1]$, group $g_j$ consists of only agents whose favorite alternative is $x_j$, while in group $g_k$ the top preference of half the agents is $a$ and the top preference of the remaining half is $b$. By Lemma~\ref{lem:lower:ordinal:top}, $x_j$ must be the representative of group $g_j, j \in [k-1]$, while either $a$ or $b$ must be the representative of group $g_k$, say $a$. So, by assumption, $a$ is also the overall winner. 
Consider the following consistent distances of the agents from the alternatives:
\begin{itemize}
    \item All agents in the first $k-1$ groups are on $a$ and report rankings consistent to the above structure (that is, the agents of group $g_j$ rank $x_j$ first even though they are at distance $1$ from all other alternatives);
    \item In group $g_k$, the $\lambda/2$ agents that prefer $a$ are at distance $1/2$ from all alternatives while the $\lambda/2$ agents that prefer $b$ are on $a$. 
\end{itemize}
Hence, 
\begin{align*}
    \SC(a) = \frac{\lambda}{4}.
\end{align*}
whereas, for any $j \in [k-1]$,
\begin{align*}
    \SC(x_j) = (k-1)\lambda\cdot 1 + \frac{\lambda}{2}\cdot \frac12 + \frac{\lambda}{2}\cdot 1 = (k-1)\lambda + \frac{3\lambda}{4}.
\end{align*}
The distortion is thus at least $4k-1$.
\end{proof}

The construction in the proof of Theorem~\ref{thm:lower:ordinal:general} requires at least $3$ alternatives so that it is possible to have $m=k+1$ when $k \geq 2$. Also, since these alternatives are equidistant, the underlying metric space cannot be a line. Nevertheless, we can show a tight lower bound of $7$ even when the metric space is a line and there are just two alternatives with known locations. 

\begin{theorem}\label{thm:ordinal:lower:3:wlog}
The distortion of any ordinal mechanism is at least $7$, even when the metric space is a line and there are only two alternatives.     
\end{theorem}

To prove Theorem~\ref{thm:ordinal:lower:3:wlog}, we follow the same methodology as we did for showing the corresponding lower bound of $3$ for full-information mechanisms. Essentially, the same structural lemmas from the previous section hold here as well but the ordinal information allows for the agents to be located at different points on the line, which in turn allows us to show a larger lower bound of $7$. 

In all instances in the following construction, there will be two alternatives located at $0$ and $1$ on the line of real numbers. Without loss of generality, we again assume that, if there are just two groups and each of them has a different representative (that is, one is represented by alternative $0$ while the other is represented by alternative $1$), then the winner is $0$. 
Since there are only two alternatives, any group can be characterized by the number of agents that prefer alternative $0$; for any $x \in [0,1]$, let $g_x$ be the group in which an $x$-fraction of the agents prefers alternative $0$. All groups will have the same size of $1$ (which is partitioned into a fraction that prefers alternative $0$ and a fraction that prefers alternative $1$), which allows us to simplify our calculations. 

\begin{lemma}\label{lem:ordinal:small-or-large-x}
Any ordinal mechanism with distortion strictly smaller than $7$ 
must choose the representative of group $g_x$ to be 
alternative $0$ for $x \in [3/4,1]$
and alternative $1$ for  $x \in [0,1/4]$.
\end{lemma}

\begin{proof}
Let $x\in [3/4,1]$; the case $x \in [0,1/4]$ is symmetric. Suppose otherwise that the representative of $g_x$ is $1$ instead of $0$. Then, in any instance consisting of copies of $g_x$ (which means that all agents are located at $x$), 
$1$ is the only representative, and thus the overall winner. Now consider the case where the agents of $g_x$ are located such that the $x$-fraction of agents that prefer $0$ are at $1$ and the remaining $(1-x)$-fraction of agents that prefer $1$ are at $1/2$. Hence, $\SC(0)=x + (1-x)/2 = (1+x)/2$ and $\SC(1) = (1-x)/2$, leading to a the distortion of at least $(1+x)/(1-x)$, which is a non-decreasing function of $x \geq 3/4$, and is thus at least $7$. 
\end{proof}

\begin{lemma}\label{lem:ordinal:half-group}
Suppose an ordinal mechanism with distortion strictly smaller than $7$ which chooses alternative $1$ as the winner when there are $2\ell+1$ groups such that $\ell$ of them are represented by $0$ while the remaining $\ell+1$ are represented by $1$, for any $\ell \geq 1$. Then, the mechanism must choose alternative $0$ as the representative of group $g_{(2\ell+3)/(4\ell+4)}$.
\end{lemma}

\begin{proof}
Suppose otherwise that $1$ is chosen as the representative of group $g_{(2\ell+3)/(4\ell+4)}$. 
Then, consider an instance consisting of $2\ell+1$ groups such that there are $\ell$ copies of $g_1$ and $\ell+1$ copies of $g_{(2\ell+3)/(4\ell+4)}$. By Lemma~\ref{lem:ordinal:small-or-large-x}, the representative of $g_1$ must be $0$, while the representative of $g_{(2\ell+3)/(4\ell+4)}$ is $1$ by assumption. Hence, we have an instance in which there are $\ell$ groups represented by $0$ and $\ell+1$ represented by $1$, leading to $1$ being chosen as the winner by the mechanism due to the assumption of the statement.

Now consider the following consistent positions for the agents on the line: 
All the agents of $g_1$ are located at $1$, 
the $(2\ell+3)/(4\ell+4)$-fraction of $g_{(2\ell+3)/(4\ell+4)}$ that prefer $0$ are located at $1$, 
and the remaining $(2\ell+1)/(4\ell+4)$-fraction of $g_{(2\ell+3)/(4\ell+4)}$ that prefer $1$ are located at $1/2$. 
Hence, 
\begin{align*}
    \SC(1) = (\ell+1) \cdot \frac{2\ell+1}{4\ell+4} \cdot \frac12  = \frac{2\ell+1}{8}
\end{align*}
whereas
\begin{align*}
    \SC(0) = \ell + (\ell+1) \cdot \frac{2\ell+3}{4\ell+4}  + (\ell+1) \cdot \frac{2\ell+1}{4\ell+4} \cdot \frac12 = \frac{14\ell+7}{8}, 
\end{align*}
leading to a distortion of $7$, a contradiction. 
\end{proof}

Using Lemmas~\ref{lem:ordinal:small-or-large-x} and~\ref{lem:ordinal:half-group}, we can now argue that the assumption of Lemma~\ref{lem:ordinal:half-group} is true for any low-distortion mechanism. 

\begin{lemma}\label{lem:ordinal:majority}
Any ordinal mechanism with distortion strictly smaller than $7$ must choose alternative $1$ as the winner when there are $2\ell+1$ groups such that $\ell$ of them are represented by $0$ while the remaining $\ell+1$ are represented by $1$, for any $\ell \geq 1$.
\end{lemma}

\begin{proof}
We will prove the statement by induction using Lemmas~\ref{lem:ordinal:small-or-large-x} and~\ref{lem:ordinal:half-group} repeatedly. 

\medskip
\noindent
{\bf Base case: $\ell=1$.}
Suppose otherwise that the mechanism chooses $0$ as the winner when there are $3$ groups such that one of them is represented by $0$ while the other two are represented by $1$. Consider an instance consisting of $g_{3/4}$ and two copies of $g_0$. By Lemma~\ref{lem:ordinal:small-or-large-x}, the representative of $g_{3/4}$ must be $0$ (using $x=3/4$), while the representative of $g_0$ must be $1$ (using $x=0$). Hence, the overall winner is $0$. Consider the following consistent metric:
The $3/4$-fraction in $g_{3/4}$ that prefers $0$ is at $1/2$, while the remaining $1/4$-fraction in $g_{1/4}$ that prefers $1$ as well as all the agents in the copies of $g_0$ are at $0$.
So, $\SC(0) = \frac34 \cdot \frac12 = 3/8$ and $\SC(1) = \frac34 \cdot \frac12 + \frac14 + 2 = 21/8$, leading to a distortion of $7$, a contradiction. 

\medskip
\noindent 
{\bf Hypothesis:} We assume that the statement is true for $\ell-1$, that is, the winner is $1$ when there are $2(\ell-1)+1 = 2\ell-1$ groups such that $\ell-1$ of them are represented by $0$ while the remaining $\ell$ are represented by $1$. Hence, by Lemma~\ref{lem:ordinal:half-group}, the mechanism must choose $0$ as the representative of group $g_{(2(\ell-1)+3)/(4(\ell-1)+4)} = g_{(2\ell+1)/(4\ell)}$.

\medskip
\noindent 
{\bf Induction step:}
Suppose otherwise that the mechanism chooses $0$ as the winner when there are $2\ell+1$ groups such that $\ell$ of them are represented by $0$ while the remaining $\ell+1$ are represented by $1$. Consider an instance consisting of 
$\ell$ copies of $g_{(2\ell+1)/(4\ell)}$ which are represented by $0$ due to our induction hypothesis, 
and $\ell+1$ copies of $g_0$ which are represented by $1$ by Lemma~\ref{lem:ordinal:small-or-large-x}.
The metric space is the following:
The $(2\ell+1)/(4\ell)$-fraction of each copy of $g_{(2\ell+1)/(4\ell)}$ that prefers $0$ is at $1/2$, while the 
remaining $(2\ell-1)/(4\ell)$-fraction of each copy of $g_{(2\ell+1)/(4\ell)}$ that prefers $1$ as well as all the agents in the copies of $g_0$ are at $0$. 
So,
\begin{align*}
    \SC(0) = \ell \cdot \frac{2\ell+1}{4\ell} \cdot \frac12 = \frac{2\ell+1}{8}
\end{align*}
and
\begin{align*}
    \SC(1) = \ell \cdot \frac{2\ell+1}{4\ell}\cdot \frac12 + \ell \cdot \frac{2\ell-1}{4\ell} + \ell+1 = \frac{14\ell+7}{4},
\end{align*}
which means that the distortion is at least $7$, a contradiction. 
\end{proof}

We are now ready to prove Theorem~\ref{thm:ordinal:lower:3:wlog}.

\begin{proof}[Proof of Theorem~\ref{thm:ordinal:lower:3:wlog}]
By Lemma~\ref{lem:ordinal:majority} and Lemma~\ref{lem:ordinal:half-group}, we have that any ordinal mechanism with distortion strictly smaller than $7$ must choose $0$ as the representative of group $g_{(2\ell+3)/(4\ell+4)}$ for any $\ell \geq 1$. Taking $\ell$ to infinity, we have that the mechanism chooses $0$ as the representative of group $g_{1/2+\varepsilon}$ for some infinitesimal $\varepsilon > 0$. 

Now consider an instance that consists of $g_{1/2+\varepsilon}$ and $g_0$. Since the former is represented by $0$ and the latter is represented by $1$ (due to Lemma~\ref{lem:ordinal:small-or-large-x}), the winner is $0$. However, the metric space might be such that 
the $(1/2+\varepsilon)$-fraction of $g_{1/2+\varepsilon}$ that prefers $0$ is at $1/2$, while the remaining $(1/2-\varepsilon)$-fraction of $g_{1/2+\varepsilon}$ that prefers $1$ as well as all the agents in $g_0$ is at $0$. Hence, $\SC(0) \approx 1/4$ and $\SC(1) \approx 1/4 + 1/2 + 1 = 7/4$, the distortion is at least $7$. 
\end{proof}

\section{Conclusion}
We studied an obnoxious distributed metric voting model and showed tight bounds on the distortion of full-information and ordinal mechanisms that operate in two steps by first choosing representatives for the groups of agents and then choose one of the representatives as the overall winner. There are many interesting directions for future research. In terms of our model, one could investigate whether improved bounds are possible for particular value combinations for $m$ (the number of alternatives) and $k$ (the number of groups), or focus on mechanisms that have access to different types of information (e.g., {\em threshold approvals}~\citep{anshelevich2024approvals}), as well as {\em randomized} mechanisms. 
Going beyond our work, it would be interesting to study the distortion of obnoxious versions of other social choice problems, such as multi-winner voting and participatory budgeting.

\subsection*{Acknowledgments}
The author would like to thank Elliot Anshelevich for fruitful discussions at early stages of this work, and the anonymous reviewers for their valuable feedback.

\bibliographystyle{plainnat}
\bibliography{references}

\end{document}